\documentclass[11pt]{article}

\usepackage[margin=1in]{geometry}

\usepackage{newpxtext}

\usepackage{amsmath}
\usepackage{amssymb}
\usepackage{amsthm}
\usepackage{mathtools}
\usepackage{enumitem}

\usepackage[dvipsnames]{xcolor}
\colorlet{LinkColor}{purple}
\colorlet{CiteColor}{Aquamarine}
\colorlet{URLColor}{magenta}
\definecolor{yqlPurple}{HTML}{8346FF}

\usepackage[
colorlinks=true,
linkcolor=LinkColor,
citecolor=CiteColor,
urlcolor=URLColor,
linktoc=page
]{hyperref}

\usepackage{newpxmath}
\usepackage[nameinlink,noabbrev]{cleveref}

\theoremstyle{plain}
\newtheorem{theorem}{Theorem}
\newtheorem{lemma}[theorem]{Lemma}
\newtheorem{proposition}[theorem]{Proposition}

\theoremstyle{definition}

\theoremstyle{remark}

\title{Half Veto, Half Maximal Lottery, Five-halves Distortion}
\author{Qilin Ye\thanks{Stanford University, \href{mailto:yql@stanford.edu}{\texttt{yql@stanford.edu}}}}

\usepackage{color-edits}
\addauthor[Qilin]{qilin}{cyan!80}

\begin{document}

\maketitle

\begin{abstract}
    We give a randomized voting rule with expected metric distortion $5/2$, improving the previous best upper bound of $2.753$. Our rule is the equal mixture of a Maximal Lottery \cite{CRWW24} and a lottery obtained by averaging scores in Simultaneous Plurality Veto \cite{KK23} over time. We further show that $5/2$ is optimal within a broader family that allows profile-dependent mixing weights and arbitrary, possibly adaptive, weights over the veto process. 
\end{abstract}

\section{Introduction}

Voting is a versatile model of collective decision making, applying in a diverse range of settings, from constituents choosing from representatives, organizations choosing from job applicants, and friends choosing from activities. The most commonly studied form of voting is  the \emph{single-winner election}, where voters provide ranked preferences over candidates, and a \emph{voting rule} selects a winning candidate. Social choice theorists aim to design voting rules satisfying a variety of different objectives which are more or less satisfiable. The classical approach, which considers all-or-nothing binary properties (like \emph{strategy-proofness}), often finds that they are unsatisfiable by reasonable voting rules (for example, see Arrow's impossibility theorem \cite{arrow1950difficulty} and the Gibbard--Satterthwaite theorem \cite{gibbard1973manipulation,satterthwaite1975strategy}). 

An alternate strategy, common in computational social choice, is to instead consider \emph{quantitative} objectives and seek to satisfy them as well as possible. Along this vein, the \emph{metric distortion} framework \cite{DBLP:conf/aaai/AnshelevichBP15,DBLP:journals/ai/AnshelevichBEPS18} has emerged as a compelling approach. Motivated by well-studied \emph{spatial} models of social choice, it imagines that the voters and candidates share a common metric space, and each voter's preference over candidates is by increasing order of distance. The goal is to design a voting rule which always chooses an \emph{approximately optimal} candidate, whose average distance to the voters is within a small factor of the minimum possible. The \emph{distortion} of a voting rule is its worst-case approximation factor, over all elections and all metrics. 

It is easy to show that every deterministic voting rule has distortion at least 3 
\cite{DBLP:conf/aaai/AnshelevichBP15,DBLP:journals/ai/AnshelevichBEPS18}, and every randomized voting rule has (expected) distortion at least 2 \cite{DBLP:journals/jair/AnshelevichP17,DBLP:conf/sigecom/FeldmanFG16}. Following a long line of work \cite{DBLP:conf/sigecom/GoelKM17,DBLP:conf/ec/MunagalaW19,DBLP:conf/aaai/Kempe20a}, \cite{DBLP:conf/focs/GkatzelisHS20} proved the existence of a deterministic voting rule with optimal distortion 3. Not long after, two works of Kizilkaya and Kempe \cite{DBLP:conf/ijcai/KizilkayaK22,KK23} introduced natural and intuitive voting rules that can be shown to have distortion 3 with simple proofs. The first was \emph{Plurality Veto} \cite{DBLP:conf/ijcai/KizilkayaK22}, which works in the following way. Each candidate has a score, which is initially the total number of voters who rank her first.

Then, following an arbitrary order, each voter decrements (\emph{vetoes}) the score of their least favorite candidate with positive score. The last candidate with positive score (or equivalently, the candidate vetoed by the last voter) wins. While this voting rule is already an elegant resolution to the problem, it is peculiar that it is not \emph{anonymous}; the particular voter order influences the outcome. Addressing this issue, \cite{KK23} proposed another natural voting rule called \emph{Simultaneous Plurality Veto}, which also has distortion 3. In this voting rule, instead of having voters decrement scores iteratively in a discrete process, each candidate's score is reduced continuously, at a rate proportional to the fraction of voters who rank it last among those with positive score.

While the optimal distortion for deterministic rules is comfortably resolved, the picture for randomized rules remains much less complete. Early on, \cite{DBLP:journals/jair/AnshelevichP17,DBLP:conf/sigecom/FeldmanFG16} showed that \textit{Random Dictatorship}, a rule that chooses the favorite candidate of a random voter, has distortion $3$, leaving a gap with the simple lower bound of $2$. Over time, the lower bound was the first to move: two independent works \cite{DBLP:conf/soda/CharikarR22,pulyassary2021randomized} showed that distortion $2$ is impossible, with the strongest constructions approaching $2.1126$. However, improving over distortion 3 proved to be a stubborn challenge, in part because certain broad classes of voting rules were known to be unable to do any better. These include voting rules that only choose from voters' top choices \cite{DBLP:conf/aaai/GrossAX17} (like Plurality Veto and Simultaneous Plurality Veto), and \emph{tournament rules} which only use aggregate comparisons between pairs of candidates \cite{DBLP:conf/sigecom/GoelKM17} (the Copeland Rule and Borda Count are two well-known examples).

Finally, the first improvement to the upper bound came from \cite{CRWW24}, who achieved distortion $\approx 2.753$. Their voting rule mixes between Maximal Lotteries, a randomized tournament rule dating to the 60's \cite{kreweras1965aggregation} derived from the equilibrium of a simple zero-sum game, and another new voting rule which conceptually is a hybrid between Random Dictatorship and  weighted version of the Copeland rule introduced by \cite{DBLP:conf/ec/MunagalaW19}. While neither of these rules alone is sufficient to beat distortion 3, they used an intricate argument to show that the worst-case instances of one rule are the best-case instances of the other, and so a carefully designed mixture of the two performs well on every instance.

In this paper, we give a simpler rule and improve the upper bound to $5/2$. We keep the maximal lottery from \cite{CRWW24} and replace the other lottery with one based on the Simultaneous Plurality Veto process \cite{KK23} with a twist. Instead of selecting only the final survivor, we average each candidate's score over time to form what we call the \textit{Veto Lottery}. We prove that the equal[!!] mixture of a maximal lottery and the veto lottery has distortion at most $5 /2$. We also prove a matching lower bound for a broader family of such rules: even if the veto scores are weighted adaptively over time, and if the mixing probability (between veto and maximal lotteries) depends on the preference profile, no rule in this family can achieve distortion below $5/2$. Thus, our simple rule has distortion exactly $5/2$ and is optimal within this family.

\subsection{Concurrent Work}

In concurrent and independent work, Frank \cite{frank2026improvedboundrandomizedmetric} considers the same rule, called \textit{Mixed Integrated Veto}, and proves that its expected metric distortion is $5/2$. Frank takes a more practical approach by discretizing the veto process based on candidate elimination; in contrast, our analysis formulates the rule in continuous time. We additionally establish a broader impossibility result: even if the veto lottery is allowed to be computed adaptively, and even if the mixture between veto and maximal lottery can be profile dependent, no rule of this form can achieve distortion below $5/2$.

\section{Preliminaries}

An election instance consists of a finite set $C$ and a population of voters, which we normalize to mass $1$. Each voter $v$ reports a strict ranking $\succ_v$ of the candidates, where $i \succ_v j$ means that $v$ prefers $i$ to $j$. We call the collection of everyone's rankings a \textit{preference profile}.

Here, we are mainly concerned with \textit{randomized} voting rules, where we select a single winner based on certain distributions. Formally, we say a \textit{lottery} $D$ is a probability distribution over $C$, where $D(i)$ is the probability assigned to candidate $i$. A randomized voting rule $f$ then takes a preference profile $\sigma$ as input, and maps it to a lottery $D$.

\paragraph{Metric Distortion.} Let $(M,d)$ be a metric space defined on the set of voters and candidates. We say $d$ is \textit{consistent} with the profile if  $d(v,i) \leqslant d(v,j)$ whenever $i \succ_v j$. Now fix a consistent metric $d$. We define \textit{social cost} of a candidate $i$ to be $SC(i) = \mathbb{E}_v d(v,i)$, and the social cost of a lottery $D$ to be $SC(D) = \sum_{i\in C}^{} D(i) SC(i)$.

The \textit{distortion} of $D$ under this metric is the ratio between $SC(D)$ and $\min_{i\in C} SC(i)$, i.e., the lowest social cost among any candidate. Correspondingly, the distortion of a randomized voting rule $f$ is the worst-case guarantee of this ratio across all preference profiles, i.e., 
\[
		\mathrm{distortion} (f) = \sup_{\sigma} \sup_{d \text{ consistent with }\sigma} \frac{SC(f(\sigma))}{\min_{i\in C} SC(i)}.
\]

For disjoint sets $I,J \subset C$, let $s_{I\succ J}$ be the share of voters who rank every candidate in $I$ above every candidate in $J$. We omit braces around singleton sets.

\paragraph{Biased Metrics.} We import several key results from \cite{CRWW24}. Fix a profile and let candidate $i^*$ be one that minimizes the social cost. A \textit{biased metric} is described by nonnegative values $x_i \geqslant 0$, with $x_{i^*} = 0$. Roughly speaking, it places voters as close to $i^*$ and as far from the other candidates relative to $i^*$ as their rankings allow.

\begin{proposition}[\cite{CRWW24}, Proposition 13]{}
		To bound worst-case distortion, it is enough to consider biased metrics only.
\end{proposition}

Now fix a biased metric. For $t\geqslant 0$, let $I_t = \{i\in C: x_i \leqslant  t\}$. Define 
\[
		r(t) = \mathbb{P}_{v} [\text{there are }a\succ_v b\text{ with }x_a - x_b > t], \qquad \ell(D,t) = \sum_{j \notin  I_t}^{} s_{I_t \succ j} D(j).
\]
Then the quantity $\ell(D,t)$ is linear in $D$. From this definition, if a voter ranks a candidate outside of $I_t$ above $i^*$, then they are counted by $r(t)$, so a handy inequality is $1 - s_{i^* \succ I_t^{c}} \leqslant  r(t)$. Later, our main argument is driven by the following statement: 

\begin{proposition}[\cite{CRWW24}, Section 3]{}
    \label{prop:distortion}
		For every lottery $D$, 
		\[
				2 SC(i^*) = \int_{0}^{\infty} r(t) \;\mathrm{d}t, \qquad SC(D) - SC(i^*) = \int_{0}^{\infty} \ell(D,t) \;\mathrm{d}t.
		\] 
		It follows that $D$ has distortion $1 + 2 \gamma$ whenever $\ell(D,t) \leqslant  \gamma \cdot r(t)$ for every $t\geqslant 0$ and every biased metric. 
\end{proposition}

\paragraph{Maximal Lotteries.} Following the convention of \cite{CRWW24}, we set $s_{i\succ i} = 1 /2$. For two lotteries $D$ and $Q$, define
\[
		s_{D\succ Q} = \sum_{i,j\in C}^{} D(i) Q(j) s_{i\succ j}.
\] 
We say a lottery $D_{\mathrm{ML} }$ is a \textit{maximal lottery} (ML) if $s_{Q\succ D_{\mathrm{ML} }} \leqslant 1 /2$ for every lottery $Q$. Intuitively, this states that if one candidate is drawn from each lottery, then a random voter would prefer the candidate drawn from $Q$ with probability at most $1 /2$. ML consists of ``half'' of our rule, and we will assume the following guarantee:  

\begin{proposition}[\cite{CRWW24}, Theorem 1]{}
    \label{prop:ml-ineq}
		For every biased metric and every $t\geqslant 0$, 
		\[
				\ell(D_{\mathrm{ML} }, t) \leqslant  \min_{} \{r(t), 1 /2\}.
		\] 
\end{proposition}
\section{In Search of a Good Rule to Accompany Maximal Lottery}

In their breakthrough, \cite{CRWW24} pairs a maximal lottery with a so-called ``random dictatorship on the $\beta$-weighted uncovered set.'' 
The parameter $\beta$ is sampled from a carefully tuned interval, and the resulting lottery is mixed with ML using an adaptive weight. These choices are highly effective in the analysis (and in achieving low distortion), but their purpose is far from easy to understand intuitively.

We would like to know whether we can replace the latter with a lottery that is simpler, whose role is easier to understand but still synergizes well with maximal lottery. This brings us back to the inequality $\ell (D_{\mathrm{ML}}, t) \leqslant \min \{r(t), 1/2\}$ in \Cref{prop:ml-ineq}.

By \Cref{prop:distortion}, we wish to make $\ell (D_{\mathrm{ML}}, t)$ small relative to $r(t)$. When $r(t)$ is large, the fixed $1/2$ becomes quite helpful, but when $r(t)$ is small, the inequality reduces to $\ell (D_{\mathrm{ML}}, t) \leqslant r(t)$, so the $1/2$ no longer helps. Thus, we want a second rule that becomes more effective when $r(t)$ is small.

By definition, for all but a share of $r(t)$ voters, every preference of form $a\succ_v b$ satisfies $x_a \leqslant x_b + t$. Intuitively, this means that when $r(t)$ is small, most voters agree on the rough direction of which candidates are better and which are worse: they may disagree about candidates with similar $x$-values, but they rarely rank a candidate with a much larger $x$-value above one with a much smaller $x$-value. Taking this observation to the extreme, it is not surprising that candidates that are \textit{much} farther from the optimum tend to be ranked below candidates that are \textit{much} closer. Consequently, we want the second rule to exploit this information, and this is where veto, a rule that specifically queries \textit{both} top- and bottom-choices, comes in. Importantly, when many voters agree that certain candidates are poor choices, their vetoes reinforce one another and get these bad candidates removed. This gives us a solid reason to try combining ML with veto.\footnotemark

Once we settle on veto, the \textit{Simultaneous Plurality Veto} process \cite{KK23} is a natural starting point. A simple way to turn their process into a lottery is to average each candidate's ``score'' over time, which we detail below.

\footnotetext{On a personal note, while I was an undergraduate at USC, Professor David Kempe was one of the first CS theoreticians I got to know, and his work on veto-based voting rules was among the first TCS research I read. Veto has therefore remained one of the first ideas I try in new settings --- sometimes successfully, but most of the time not.}

\section{Veto Lottery $+$ Maximal Lottery $=$ Distortion $5 /2$}

We first define the other ``half'' of our rule, which we call the \textit{Veto Lottery}. We draw inspiration from the simultaneous veto process of \cite{KK23}. 

Formally, let $\mathrm{plu} (i)$ be the share of voters who rank candidate $i$ as their first choice, and let $\mathrm{plu} (I) = \sum_{i\in I}^{} \mathrm{plu} (i)$. We start by giving each candidate $i$ a score $z_i(0) = \mathrm{plu} (i)$ and decrease the candidates' scores over time. We say a candidate is \textit{active} while her score is positive. At each time $t$, every voter lowers the score of her least favorite active candidate at unit rate. If $\delta_i(t)$ is the share of voters lowering candidate $i$ at time $t$, then $z_i'(t) = -\delta_i(t)$ while $i$ is active. Since the total score falls at rate $1$, we have 
\[
		\sum_{i\in C}^{} z_i(t) = 1-t, \qquad 0\leqslant t\leqslant 1.
\]

We define the \textbf{veto lottery} $D_V$ by averaging these scores over time:
\[
		D_V(i) = 2 \int_{0}^{1} z_i(t) \;\mathrm{d}t.
\] 
Observe that $D_V(i) \geqslant 0$, and $\sum_{i\in C}^{} D_V(i) = 2 \int_{0}^{1} (1-t) \;\mathrm{d}t = 1$, so indeed, $D_V$ is a lottery.
\vspace{5pt}

We are now ready to state the rule: let $D_{\mathrm{ML} }$ be any maximal lottery, and let $D_V$ be the veto lottery. Our rule simply concerns $D^* = (D_V + D_{\mathrm{ML} }) /2$.

\begin{theorem}{}{}
		\label{thm:ub}
		The distortion of $D^*$ is at most $5 /2$.
\end{theorem}

We first prove a useful lemma. 

\begin{lemma}{}{}
		\label{lem:ub}
		For every nonempty proper subset $I \subset C$, with $J = I^{c}$, 
		\[
				\sum_{j\in J}^{} s_{I \succ j}D_V(j) \leqslant \mathrm{plu} (J)^2.
		\] 
\end{lemma}

\begin{proof}{}{}
		We first extend the notation of $z$ and $\delta$ to define 
		\[
				z_J(t) = \sum_{j\in J}^{} z_j(t), \qquad \delta_J(t) = \sum_{j\in J}^{} \delta_j(t).
		\] 

		Fix any moment in the process, and suppose $j\in J$ is active. Every voter counted by counted by $s_{I\succ j}$ ranks every candidate in $I$ above $j$. Since $j$ is itself active,such a voter's least preferred active candidate cannot lie in $I$, and hence must lie in $J$. 
        Hence $s_{I\succ j} \leqslant \delta_J(t)$ whenever $z_j(t) > 0$. Also, $z'_J(t) = -\delta_J(t)$. Therefore, 
		\begin{align*}
				\sum_{j\in J}^{} s_{I\succ j}D_V(j) &= 2 \int_{0}^{1} \sum_{j\in J}^{} s_{I\succ j} \cdot z_j(t) \;\mathrm{d}t \\
				&\leqslant  2 \int_{0}^{1} \delta_J(t) \cdot z_J(t) \;\mathrm{d}t \\
				&= - \int_{0}^{1} \frac{\mathrm{d}}{\mathrm{d}t} z_J(t)^2 \;\mathrm{d}t \\
				&=  z_J(0)^2 = \mathrm{plu} (J)^2. \qedhere
		\end{align*}
\end{proof}

\begin{proof}[Proof of \Cref{thm:ub}]{}
		We now conclude the  main proof. Fix $t\geqslant 0$. Since the optimal candidate $i^*$ lies in $I_t$, 
		\[
				\mathrm{plu} (I_t^{c}) \leqslant 1 - s_{i^*\succ I_t^{c}} \leqslant r(t).
		\] 
		If $I_t = C$, then $\ell(D_V, t) = 0$. Otherwise, applying \Cref{lem:ub} with $I = I_t$ gives
		\[
				\ell(D_V, t) \leqslant \mathrm{plu} (I_t^{c})^2 \leqslant r(t)^2.
		\] 

		On the other hand, for the maximal lottery, recall from \Cref{prop:ml-ineq} that $\ell(D_{\mathrm{ML} }, t) \leqslant \min_{} \{r(t), 1 /2\}$. Since $\ell$ is linear in the lottery, 
		\[
				\ell(D^*, t) \leqslant \frac{r(t)^2 + \min_{} \{r(t), 1 /2\}}{2}.
		\] 
		For every $r\in [0,1]$, we have $(r^2 + \min_{} \{r, 1 /2\}) /2 \leqslant 3r /4$. Therefore, $\ell(D^*, t) \leqslant 3 /4 \cdot r(t)$ for every $t$, and so $\mathrm{distortion} (D^*) \leqslant 1 + 2 \cdot 3 /4 = 5 /2$, as claimed.
\end{proof}

\section{Veto Lottery $+$ Maximal Lottery Cannot Do Better}

So far, the rule we have considered is surprisingly simple: we take a simple average between the veto lottery and a maximal lottery, and even the veto lottery itself is obtained from a simple average over time. This naturally raises the question of whether more careful choices can improve the $5 /2$ guarantee.

In this section, we answer that in the negative. Specifically, we allow two things to vary:

\begin{enumerate}[label=(\arabic*),align=left]
		\item \textit{The mixture weight}. We choose any $\lambda\in [0,1]$, possibly profile-dependent, and mix the veto lottery with a maximal lottery with weights $\lambda$ and $1-\lambda$, respectively.
		\item \textit{The veto time weights}. At each time $t<1$, we normalize the veto scores to $P_t(i) = z_i(t) / (1-t)$. For each preference profile, we choose any (possibly adaptive) probability measure $\mu$ on $[0,1)$ and define
				\[
						V(i) = \int_{[0,1)}^{} P_t(i) \;\mathrm{d}\mu(t).
				\] 
\end{enumerate}
Given these choices, define 
\[
		D_{\lambda, V} = \lambda V + (1-\lambda) D_{\mathrm{ML}}.
\] 
It is easy to verify that our original rule is the canonical case with $\lambda = 1 /2$ and $\mathrm{d} \mu(t) = 2(1-t) \mathrm{d} t$, since in this case
\[
		V(i) = \int_{0}^{1} \frac{z_i(t)}{1-t} \cdot 2(1-t) \;\mathrm{d}t = 2 \int_{0}^{1} z_i(t) \;\mathrm{d}t = D_V(i).
\] 
We will show that even with these extra choices, the worst-case distortion cannot be improved.

\begin{theorem}{}{}
		\label{thm:lb}
		For every choice of $V$ and $\lambda\in [0,1]$, even if they are profile-dependent/adaptive, the distortion of $D_{\lambda, V}$ is at least $5 /2$. Consequently, the mixture in its simplest form--simple average and time-averaged veto--cannot be improved. 
\end{theorem}

\begin{proof}{}{}
		We use one family of preference profiles and two different metrics on the same profile. Our first metric makes the maximal lottery expensive, while the other makes every veto average expensive.

		We first describe the preference profile:
		\begin{enumerate}[label=(\arabic*),align=left]
				\item Let $\alpha \in (1 /2, 1)$ and $m\geqslant 3$.
				\item Let the candidates be $O$ (optimal), $M$ (ML winner), $X_1, \hdots, X_m$, and $Y_1, \hdots, Y_m$.
				\item Let there be $2m$ voter groups, defined as follows, where indices are interpreted cyclically:
						\begin{itemize}
								\item Group $X_i$, of mass $\alpha / m$, has the following ranking:
										\[
												X_i \succ M \succ O \succ X_{i+1} \succ \hdots \succ X_{i-1} \succ Y_{i+1} \succ \hdots \succ Y_i.
										\] 
								\item Group $Y_i$, of mass $(1-\alpha) /m$, has the following ranking:
										\[
												Y_i \succ O \succ Y_{i+1} \succ \hdots \succ Y_{i-1} \succ M \succ X_{i+1} \succ \hdots X_i.
										\] 
						\end{itemize}
		\end{enumerate}

		Immediately, candidate $M$ beats $O$ and every $Y_i$ by a share $\alpha > 1 /2$, and $M$ beats every $X_i$ by a share $1 - \alpha / m > 1 /2$. Therefore, ML is the point mass on $M$.

		Now consider the veto process, which only concerns the $X_i$'s and $Y_i$'s as $O$ and $M$ receive zero plurality score. Until the $Y_i$'s are eliminated, each $X_i$ is vetoed at rate $(1-\alpha) /m$, while each $Y_i$ is vetoed at rate $\alpha /m$. Hence, all $Y_i$'s are eliminated at time $(1-\alpha) / \alpha$, while the $X_i$'s still have positive score because $\alpha > 1 /2$. 

		Write $P_t(X) = \sum_{i=1}^{m} P_t(X_i)$. Before the $Y_i$'s are eliminated, $P_t(X) = (\alpha - (1-\alpha)t) / (1-t) \geqslant \alpha$, and afterwards $P_t(X) = 1$. Thus $P_t(X) \geqslant \alpha$ for every $t < 1$. It follows immediately that every veto average satisfies
		\[
				V(X) := \sum_{i=1}^{m} V(X_i) \geqslant \alpha,
		\] 
		and this holds for every choice of $\mu$, adaptive or not.

		We now describe the first metric, in which ML approaches distortion $3$, and veto approaches $2$. In this metric, we place $O$ and all $Y_i$'s at $0$, and $M$ and all $X_i$'s at $1$. Place every $X$-voter at $1 /2$ and every $Y$-voter at $0$. Then the social costs are
		\[
				SC(O) = SC(Y_i) = \frac{\alpha}{2}, \qquad SC(M) = SC(X_i) = 1-\frac{\alpha}{2}.
		\] 
		Hence the distortion of ML is $SC(M) / SC(O) = 2 /\alpha - 1$. The veto average is supported on the $X_i$'s and $Y_i$'s. Since $V(X) \geqslant \alpha$, its distortion is at least
		\[
				\alpha \left( \frac{2}{\alpha}-1 \right)  + (1-\alpha) = 3-2\alpha.
		\] 
		Therefore, $D_{\lambda, V}$ has distortion at least
		\[
				\lambda (3-2\alpha) + (1-\lambda) \left( \frac{2}{\alpha}-1 \right) 
		\] 
		which tends to $3-\lambda$ as $\alpha \to  1 /2$.

		We now describe the second metric, in which every veto average approaches distortion $3$, while ML approaches $2$. Consider a shortest-path metric, induced by the following. First make a complete bipartite graph between the voters and the candidates. Each voter has an edge of length $1$ to $O$ and everyone she ranks above $O$, and an edge of length $3$ to every candidate she ranks below $O$. In this metric, the social costs are
		\begin{align*}
				SC(O) = 1, &\qquad SC(M) = 3 - 2\alpha \\
				SC(X_i) = 3 - \frac{2\alpha}{m}, &\qquad SC(Y_i) = 3 - \frac{2(1-\alpha)}{m}.
		\end{align*}

		Thus the ML has distortion $3 - 2\alpha$. Moreover, since $\alpha > 1 /2$, every $X_i$ and $Y_i$ has distortion at least $3 - 2\alpha / m$. Because every veto average is supported only on the $X_i$'s and $Y_i$'s, its distortion is also at least $3 - 2\alpha / m$. Hence, $D_{\lambda, V}$ has distortion at least
		\[
				\lambda \left( 3 - \frac{2\alpha}{m} \right)  + (1-\lambda) (3-2\alpha)
		\] 
		which tends to $2+\lambda$ as $m\to \infty$ and $\alpha \to 1 /2$. This completes the proof. 
\end{proof}

\section{Discussion}

We conclude with a brief discussion of some of the interesting aspects of our results in the context of prior work. 

First and foremost, our Veto Lottery rule extends the rich family of voting rules introduced by Kizilkaya and Kempe \cite{DBLP:conf/ijcai/KizilkayaK22,KK23}, and brings this family into the conversation of the best possible randomized voting rules. In fact, Kizilkaya and Kempe had already considered adapting Plurality Veto with randomization, by considering the class of rules which stops the veto process at an intermediate point, and chooses a random candidate with probability proportional to their scores. This class interpolates between Random Dictatorship (no veto process) and Plurality Veto (complete veto process), and in \cite{DBLP:conf/ijcai/KizilkayaK22} they showed that every rule in between has distortion 3. 

\cite{CRWW24} also suggests that optimal deterministic voting rules could be useful in improving the randomized distortion bounds, but the challenge is in finding the right variant that can take full advantage of the biased-metric framework. It is exciting that doing so is possible with Veto Lottery, and that it shares the same simplicity of its forebearers, both in its description and analysis.

Taking a step back, our result also offers an exciting perspective on the landscape of randomized metric distortion.
At first glance, the fact that most well-studied voting rules are either tournament rules or plurality-based rules, and neither can beat distortion 3, feels like evidence of a major missing piece in voting theory.
But it now turns out that these two classes of voting rules are somehow complementary, and we can get distortion as low as $5/2$ by mixing one of each in equal proportion. It is redeeming to see them come together and break the barriers they faced individually as two sides of a coin.

\section*{Acknowledgments}

The author especially thanks Prasanna Ramakrishnan for doing an \textit{amazing} job polishing the introduction and for pointing out additional connections to prior literature. The author also thanks, in alphabetical order, Ziyi Cai, Moses Charikar, Jabari Hastings, Prasanna Ramakrishnan, and Kangning Wang for various helpful discussions on a preliminary version of this paper.

\bibliographystyle{alpha}
\bibliography{references}

\end{document}